\documentclass[11pt,a4paper]{article}

\usepackage{amsfonts,amsmath,amssymb,amsthm}
\usepackage{latexsym,amscd}
\usepackage{amsbsy}
\usepackage{CJK}
\usepackage{indentfirst,mathrsfs}
\usepackage{latexsym,amscd}
\usepackage{amsbsy}
\usepackage[noadjust]{cite}
\usepackage{color}
\usepackage{multirow}
\usepackage{makecell}
\usepackage{float}
\usepackage{booktabs}
\newtheorem{theorem}{Theorem}

\newtheorem{lem}{Lemma}

\newtheorem{prop}{Proposition}

\newtheorem{example}{Example}

\newcommand{\bF}{{\mathbb F}}

\newcommand{\fqm}{{\mathbb F}_{q^m}}

\newcommand{\B}{{\mathcal{B}}}
\newcommand{\C}{{\mathcal{C}}}

\begin{document}

\title{Generalized Hamming Weights of Linear Codes from \\ Quadratic Forms over Finite Fields of Even Characteristic}

\author{Chao Liu, Dabin Zheng and Xiaoqiang Wang
  \thanks{Chao Liu, Dabin Zheng and Xiaoqiang Wang are with the Hubei Key Laboratory of Applied Mathematics, Faculty of Mathematics and Statistics, Hubei University, Wuhan 430062, China, E-mail: chliuu@163.com, dzheng@hubu.edu.cn, waxiqq@163.com. The corresponding author is Dabin Zheng. The research of Dabin Zheng and Xiaoqiang Wang was supported by NSFC under Grant Numbers 62272148, 11971156, 12001175.}}

\date{}
\maketitle

\begin{abstract}
The generalized Hamming weight of linear codes is a natural generalization of the minimum Hamming distance. They convey the structural information of a linear code and determine its performance in various applications, and have become one of important research topics in coding theory. Recently, Li (IEEE Trans. Inf. Theory, 67(1): 124-129, 2021) and Li and Li (Discrete Math., 345: 112718, 2022) obtained the complete weight hierarchy of linear codes from a quadratic form over a finite field of odd characteristic by analysis of the solutions of the restricted quadratic equation in its subspace. In this paper, we further determine the complete weight hierarchy of linear codes from a quadratic form over a finite field of even characteristic by carefully studying the behavior of the quadratic form on the subspaces of this field and its dual space, and complement the results of Li and Li.
\end{abstract}


\vskip 30pt


\section{Introduction}\label{sec1}
Let $q$ be a power of a prime number and ${\mathbb F}_q$ be a finite field with $q$ elements. A linear code $\C$ over $\bF_q$ with parameters $[n,k,d] $ is a $k$-dimensional subspace of $\bF_q^n$
with length $n$ and minimum Hamming distance $d$. For each $r$ with $1\leq r\leq k$, let $[\C, r]$ be the set of all $r$-dimensional $\bF_q$-subspaces of $\C$.
For each $H_r\in [\C, r]$, the {\em support} of $H_r$, denoted by supp($H_r$), is the set of nonzero coordinates of nonzero vectors in $H_r$, i.e.,
$$ {\rm supp}(H_r)=\left\{ i \, :\, 1\leq i\leq n, \,\, c_i\neq 0  \,\, {\rm for \,\, some} \,\, (c_1,c_2,\cdots,c_n)\in H_r \right\}.$$
The $r$-th {\em generalized Hamming weight} (GHW) of $\C$ is defined by
	$$ d_r(\C)= {\rm min} \left\{ \,|{\rm supp}(H_r)|\,:\, H_r \in [ \C, r]\, \right\},$$
where $|{\rm supp}(H_r)|$ denotes the cardinality of the set ${\rm supp}(H_r)$. The set $\left\{ d_1(\C), d_2(\C),\cdots, d_k(\C)\right\}$ is called the {\em weight hierarchy}
of $\C$. Note that $d_1(\C)$ is just the minimum Hamming distance of~$\C$.

The notion of GHWs was introduced by Helleseth et al.~\cite{Hellesethetal1977} and Kl{\o}ve~\cite{Klove1978}, which is a natural generalization of the minimum distance $d_1(\C)$. The GHWs of linear codes provide fundamental information of linear codes which are important in many applications. In 1991, Wei~\cite{Wei1991} first gave a series of beautiful results on GHW and used it to characterize the cryptography performance of a linear code over the wire-tap channel of type II. The GHW was also used to deal with t-resilient functions, and trellis or branch complexity of linear codes~\cite{Chor1985,Geometric}.  Apart from these cryptographic applications, the GHWs also can be applied to computation of the state and branch complexity profiles of linear codes~\cite{Forney1994,Kasamietal1993}, indication of efficient ways to shorten linear codes~\cite{Hellesethetal19951}, determination of the erasure list-decodability of linear codes~\cite{Guruswami2003}, etc..

The study of GHWs of linear codes has attracted much attention in the past two decades, and many results have been obtained in the literature.
For example, general lower and upper bounds on GHWs of linear codes were derived~\cite{Ashikhmin1999,Cohen1994,Hellesethetal19951}, and the GHWs have been determined or estimated for many classes of linear codes such as Hamming codes~\cite{Wei1991}, Reed-Muller codes~\cite{Heijnen1998,Wei1991}, binary Kasami codes~\cite{Hellesethetal19952}, Melas and dual Melas codes~\cite{vandergeer1994}, some BCH codes and their duals~\cite{Cheng1997,Duursma1996,Feng1992,Moreno1998,Shim1995,vandergeer19942,vandergeer1995}, some trace codes~\cite{Cherdien2001,Stichtenoth1994,vandergeer19952,vandergeer1996}, some algebraic geometry codes \cite{Yang1994} and cyclic codes \cite{Lishuxing2017,Ge2016,Yang2015}. Readers may refer to the excellent textbook~\cite{Huffman2003} for a brief introduction to GHWs, and to \cite{Geometric} for a comprehensive survey of GHWs via a geometric approach. However, to determine the weight hierarchy of linear codes is a difficult problem, and to the best of our knowledge, only a few linear codes have known complete weight hierarchies so far.

Let ${\rm Tr}$ denote the trace function from ${\mathbb F}_{q^m}$ to ${\mathbb F}_{q}$. For a set $D=\{d_1,d_2,\cdots, d_n\}\subset {\mathbb F}_{q^m}$,
we define a linear code $\C_D$ of length $n$ over $\bF_q$ as follows:
\begin{equation}\label{eq:def-code}
\C_D=\left\{ {\bf c}(x) = \left( {\rm Tr}( xd_1), {\rm Tr}(xd_2), \cdots, {\rm Tr}(xd_n) \right), \,\, x\in {\mathbb F}_{q^m} \right\}.
\end{equation}
The set $D$ is called the defining set of $\C_D$.  This method of construction of linear codes was first proposed by Ding and Niederreiter~\cite{Ding2007}.
A number of linear codes with a few weights were obtained by properly choosing defining sets. For examples see \cite{Ding2015,LiMesnager2020,Wangzhengding2021}
and reference therein.

Let $\C$ be an $[n, m]$ linear code over $\bF_q$ with generator matrix $G$. For an integer $r$ with $1\leq r\leq m$, let $U$ be an $(m-r)$-dimensional subspace of $\bF_q^m$
and $m(U)$ denote the total number of occurrences of the vectors of $U$ as columns of $G$. Helleseth et al. in \cite{Hell1992} gave a formula to calculate the $r$-th generalized Hamming weight of $\C$ as follows:
$$d_r(\C)=n-{\rm max} \left\{ \,m(U)\,:\,U\in \left[ {\mathbb F}_q^m, m-r \right] \right\}, $$
where $\left[ {\mathbb F}_q^m, m-r \right]$ denotes the set of all $(m-r)$-dimensional subspaces of $\bF_q^m$. It has been shown in~\cite{Xiangc} that all linear codes can be obtained from the defining-set construction as in~(\ref{eq:def-code}). Following the idea of Helleseth et al., Li~\cite{Li2018} gave a formula to calculate the $r$-th generalized Hamming weight of $\C_D$ as follows:
\begin{lem}\cite[Theorem~1]{Li2018}\label{lem:GHW}
For each $r$ with $0\leq r\leq m$, if the dimension of $\C_{D}$ given in (\ref{eq:def-code}) is $m$, then
	$$ d_r(\C_D)= n-{\rm max} \left\{ \, |D\cap H_r^\perp|\, :\, H_r\in [{\mathbb F}_{q^m},r] \right\},$$
where $[{\mathbb F}_{q^m},r]$ denotes the set of all $r$-dimensional $\bF_q$-subspaces of $\fqm$ and $H_r^\perp$ is the dual space of $H_r$.
\end{lem}

Let $f(x)$ be a quadratic form over $\fqm$. Define a subset of $\fqm$ as follows:
\begin{equation}\label{eq:definingset}
	D_f=\left\{ x\in \bF_{q^m} \, :\, f(x) = a,\,\,a\in \bF_q\right\}.
\end{equation}
In the case $a=0$, the weight hierarchy of $\C_{D_f}$ can be derived from the results in~\cite{Wan1997,Wan1994}, in which the discussed linear codes were represented by
generator matrices, and their weight hierarchies were deduced by application of the theory of finite projective geometry. When $a\neq 0$ and $q$ is an odd prime, Li and Li~\cite{Li2021,Li2022} used a different method to determine the weight hierarchy of the linear code $\C_{D_f}$ for $f$ being non-degenerate and degenerate, respectively.
Continuing the work of \cite{Li2021} and \cite{Li2022}, this paper further discusses the weight hierarchy of the linear code $\C_{D_f}$ for $a\neq 0$ and $q$ being a power of $2$. By Lemma~\ref{lem:GHW}, the key to obtaining the $r$-th generalized Hamming weight of $\C_{D_f}$ is to determine the maximum value of $|D_f\cap H_r^\perp|$ for a quadratic form $f$ and any $r$-dimensional subspace $H_r$ of $\bF_{q^m}$. It has been shown that $|D_f\cap H_r^\perp|$ is equal to the number of solutions for the quadratic equation $f|_{H_r^\perp}(x)=a$, where $f|_{H_r^\perp}$ is the restriction of $f$ to $H_r^\perp$. The rank and the stand type of the restriction $f|_{H_r^\perp}$
need to be determined for solving the equation $f|_{H_r^\perp}(x)=a$. This is different from and more difficult than the case of finite fields of odd characteristic. By application of quadratic form theory on finite fields of even characteristic and anatomization of the rank of $f|_{H_r^{\perp}}$ and type $f|_{H_r^{\perp}}$, the maximum values of $|D_f\cap H_r^{\perp}|$ are determined for all $r$-dimensional subspaces $H_r^\perp$ of $\bF_{q^m}$, where $1\leq r\leq m$. So, we obtain the weight hierarchy of the linear code $\C_{D_f}$ for any quadratic form $f$ over finite fields of even characteristic, and complement the results of~\cite{Li2021,Li2022}.

The rest of this paper is organized as follows: In Section 2, we introduce quadratic forms over a finite field $\bF_{q^m}$ of even characteristic and their restrictions to the subspaces of $\bF_{q^m}$. Section~3 determines the weight hierarchies of the binary linear codes $\C_{D_f}$ for $f$ being a non-degenerate quadratic form over $\fqm$.
In Section~4, we obtain the weight hierarchies of the binary linear codes $\C_{D_f}$ for $f$ being a degenerate quadratic form over $\fqm$. Finally, Section 5 concludes
the paper.

\section{Quadratic forms on ${\mathbb F}_{q^m}$ and their restrictions to its subspace}\label{sec2}

In this section, we recall some definitions and properties of quadratic forms over finite fields of even characteristic and their restrictions to subspaces.
Reader refer to~\cite{Hou2018,Li2021,Wan2003} for more details.

From now on, let $q$ be a power of $2$ and $\fqm$ be a finite field of $q^m$ elements. A polynomial $F(x)\in \fqm[x]$ with the following shape,
 \[ F(x)=\sum_{i,j=0}^{m-1} a_{ij} x^{q^i+q^j}, \,\, a_{ij}\in \fqm ,\]
is called a Dembowski-Ostrom (DO) polynomial~\cite{DembowskiOstrom1968}. It is clear that $F(x)$ is also a homogeneous quadratic polynomial. A function $Q(x_1, x_2, \ldots, x_m)$ from $\bF_{q}^m$ to $\bF_q$ is called a quadratic form
if it is a homogenous polynomial of degree two as follows:
\[Q(x_1, x_2, \cdots, x_m) = \sum_{1\leq i\leq j\leq m} a_{ij} x_ix_j , \,\, \, a_{ij}\in \bF_{q}.\]
Let ${\rm Tr}(\cdot)$ denote the trace function from $\bF_{q^m}$ to $\bF_q$. We fix a basis $\Omega=\{\varepsilon_1, \varepsilon_2, \cdots, \varepsilon_m\}\ {\rm of}\ \fqm $ over $\mathbb{F}_q$ and identify
$x =\sum_{i=1}^m x_i\varepsilon_i\in \bF_{q^m}$ with the vector ${\bf x} =(x_1,x_2,\cdots,x_m)\in \bF_q^m$, then ${\rm Tr} ( F(x))$ is
a quadratic form in the coordinates of $\bF_q^m$. Moreover, every quadratic form $f(x)$ from $\fqm$ to $\bF_q$ can be represented as
\[ f(x) = {\rm Tr} \left(F(x)\right), \]
where $F(x)$ is a DO polynomial defined above.

Let $f$ be a quadratic form on $\fqm$ and $\ell_f$ be the symmetric bilinear form on $\fqm$ associated with $f$ as follows:
$$\ell_f(x,y)=f(x+y)-f(x)-f(y), \,\,\, x, y \in \fqm .$$
By the property $f(ax)=a^2 f(x)$ for all $x\in \fqm$, where $a\in \bF_q$, and the property of the symmetric bilinear forms,  we have that
\begin{equation}\label{eq:quadraticformf}
f(x)=\sum_{i=1}^m f(\varepsilon_i) x_i^2+\sum_{1\leq i<j\leq m} \ell_f(\varepsilon_i,\varepsilon_j) x_i x_j= {\bf x}M_f(\Omega) {\bf x}^{T},
\end{equation}
where ${\bf x} =(x_1, x_2, \cdots, x_m)\in {\mathbb F}_q^m$ and
\[ M_f(\Omega) =\left( \begin{array}{cccc}
f(\varepsilon_1) & \ell_f(\varepsilon_1, \varepsilon_2) & \cdots &  \ell_f(\varepsilon_1, \varepsilon_m) \\
0 &  f(\varepsilon_2) & \cdots & \ell_f(\varepsilon_2, \varepsilon_m) \\
 \vdots & \vdots & \ddots & \vdots\\
  0 & 0 & \cdots & f(\varepsilon_m) \\
  \end{array}
\right)\]
is called the matrix of quadratic form $f$ with respect to the basis $\Omega$.
Recall that the kernel of $\ell_f$ is defined to be
\[
 {\rm ker} \, \ell_f =\left\{ x\in \fqm \,:\, \ell_f(x,y)=0\ {\rm for \,\, all}\, \, y\in \fqm \right\} = \left\{ {\bf x} \in \bF_q^m \, :\, {\bf x} \left( M_f(\Omega)+M_f^{T}(\Omega)\right)=0 \right\},
\]
and the kernel of $f$ is defined to be
\begin{align*}
	{\rm ker}\, f &=\{x\in \fqm \,:\, f(x+y)=f(y)\ {\rm for\,\, all}\,\, y\in \fqm \}\\
	&=\{x\in \fqm \,:\, f(x)=0\ {\rm and}\ l_f(x,y)=0\ {\rm for\,\, all}\,\, y\in \fqm\}.
\end{align*}
The quadratic form $f$ is said to be non-degenerate if ${\rm ker}\,f =\{0\}$. Otherwise, $f$ is degenerate.  The rank of the quadratic form $f$ over $\fqm$ is defined as
\begin{equation*}\label{eq:rankoff}
{\rm rank} \, f= m- {\rm dim}_{\bF_q} ( {\rm ker}\, f).	
\end{equation*}
It is known that ${\rm ker}\, f \subseteq {\rm ker}\, \ell_f$~\cite{Hou2018}. The type of $f$ is defined as
$${\rm type}\,f = {\rm dim}_{\bF_q}({\rm ker}\, \ell_f/{\rm ker}\, f).$$
It has been shown in Lemma~6.34 of \cite{Hou2018} that ${\rm type}\,f=0$ if $f({\rm ker}\, \ell_f)=\{0\}$, otherwise, ${\rm type}\,f=1$.
The rank of $f$ can be represented as
\begin{equation}\label{eq:rankf}
{\rm rank}\, f=m-{\rm dim}_{\bF_q}\,{\rm ker}\, \ell_f+{\rm type}\, f ={\rm rank}\left( M_f(\Omega)+M_f^{T}(\Omega) \right)+{\rm type}\, f.
\end{equation}
Taking a nonsingular linear transformation ${\bf x}={\bf y}P$, where ${\bf x, y} \in \bF_q^m$ and $P$ is an $m\times m$ nonsingular matrix over $\bF_q$, the quadratic form
$f(x)$ in (\ref{eq:quadraticformf}) is equivalent to the following standard types.

\begin{lem}\cite[Proposition 2.4]{Klapper1993}\label{lem:normalform}
Let $f$ be a quadratic form given in (\ref{eq:quadraticformf}) with rank $t$. If $t$ is even, then $f$ is equivalent under a change coordinates to
\begin{description}
\item[Tpye I:] \,\, $x_1x_2+x_3x_4+\cdots+x_{t-1}x_t$, or
\item[Type II:] \,\, $x_1x_2+x_3x_4+\cdots+ x_{t-3}x_{t-2}+\alpha x_{t-1}^2+x_{t-1}x_{t}+\alpha x_{t}^2,\,\,\alpha\in \bF_q\setminus \left\{ x^2+x \,: \, x\in \bF_q \right\}$.
\end{description}
If $t$ is odd, then $f$ is equivalent under a change coordinates to
\begin{description}
\item[Type III:]\,\, $x_1x_2+x_3x_4+\cdots+x_{t-2}x_{t-1}+x_t^2$.
\end{description}
For $a\in \bF_q$, let $\delta(a)=-1$ if $a\neq 0$ and $\delta(0)=q-1$. The number of solutions in $\fqm$ to the equation $f(x)=a$ is
\begin{description}
\item[Tpye I:] \,\,$q^{m-1}+\delta(a)q^{m-\frac{t+2}{2}}$;
\item[Type II:] \,\, $q^{m-1}-\delta(a)q^{m-\frac{t+2}{2}}$;
\item[Type III:]\,\, $q^{m-1}$.
\end{description}
\end{lem}

It is known that the rank of $M_f(\Omega)+M_f^T(\Omega)$ is even. If $f$ is equivalent to Type I or Type II, then the rank of $f$ is an even number. By (\ref{eq:rankf}) we have that ${\rm type}\, f=0$. If $f$ is equivalent to Type III, then the rank of $f$ is an odd number. By (\ref{eq:rankf}) we have that ${\rm type}\, f=1$.

Next, we recall some results on the restriction of the quadratic form $f$ to a subspace of $\fqm$. Let $H$ be an $r$-dimensional subspace of $\fqm$. The restriction of $f$ to
$H$, denoted by $f|_{H}$, is a quadratic form over $H$ with at most $r$ variables. Wan in Lemma~2 of~\cite{Wan1994} proved that the size of the intersection of a preimage of a quadratic form $f$ and a subspace $H$ is equal to the number of solutions for the restricted equation $f|_H(x)=a$. Below we give the finite-field version of Lemma 2 in~\cite{Wan1994}.

\begin{prop}\cite[Lemma 2]{Wan1994}\label{lem:intersection}
Let $f$ be a quadratic form over $\bF_{q^m}$ defined in (\ref{eq:quadraticformf}) and $D_f$ be a set given in~(\ref{eq:definingset}). Let $H$ be a $d$-dimensional subspace of $\fqm$ and $f|_{H}$ be the restriction of $f$ to $H$. Then
$$|D_f \cap H|=|D_{f|_{H}}|,$$
where $D_{f|_{H}}=\left\{  y\in H \,:\, f|_H(y)=a \right\}$.
\end{prop}

\begin{proof}
Let $\Omega=\{\varepsilon_1, \varepsilon_2, \cdots, \varepsilon_m\}$ be a basis of $\fqm $ over $\mathbb{F}_q$. Each element $x\in \fqm$ corresponds
one to one with its coordinate vector ${\bf x} = (x_1, x_2, \cdots,x_m)$ under this basis. Let $\{ \alpha_1, \alpha_2, \cdots, \alpha_d\}$ be a basis of $H$ over $\bF_q$, then
$\alpha_i = \sum_{j=1}^m b_{ij} \varepsilon_j$, where $b_{ij}\in \bF_q$ and $i=1, 2, \cdots, d$. For any $y\in H$, it can be represented as
\begin{equation}\label{eq:y}
y = \sum\limits_{i=1}^d y_i \alpha_i = \sum\limits_{j =1}^m \left(  \sum\limits_{i=1}^d b_{ij} y_i\right) \varepsilon_j,\,\,\, y_i \in \bF_q .
\end{equation}
By a discussion similar to (\ref{eq:quadraticformf}) we have
\[
f|_H(y) = \left( \sum\limits_{i=1}^d b_{i1} y_i, \cdots,  \sum\limits_{i=1}^d b_{im} y_i \right) M_f(\Omega) \left( \sum\limits_{i=1}^d b_{i1} y_i, \cdots,  \sum\limits_{i=1}^d b_{im} y_i \right)^{T}
= {\bf y}PM_f(\Omega)P^T {\bf y}^T, \]
where ${\bf y} = (y_1, y_2,\cdots, y_d)$, $P=({\bf b_1, b_2, \cdots, b_d})$ and ${\bf b_i} = (b_{i1}, b_{i2}, \cdots, b_{id})^T$ for $i=1, 2, \cdots, d$. This shows that
the matrix of the restricted quadratic form $f|_H$ is $PM_f(\Omega)P^T$.

For any $x\in D_f\cap H$, we have $f(x)={\bf x}M_f(\Omega){\bf x}^T=a$ and ${\bf x}={\bf y} P$ from (\ref{eq:y}). Then, we derive that ${\bf y}PM_f(\Omega)P^T {\bf y}^T =a$,
i.e., $f|_H(y) =a$. So, from an $x\in D_f\cap H$ we can obtain a $y\in D_{f|_H}$. Conversely, for any $y\in D_{f|_H}$, we have $f|_H(y)={\bf y}PM_f(\Omega)P^T {\bf y}^T =a$. Let ${\bf x}={\bf y}P=(x_1, x_2, \cdots, x_m)$, then ${\bf x}M_f(\Omega){\bf x}^T=a$ and $x=\sum_{i=1}^m x_i\varepsilon_i \in H$. So, from a $y\in D_{f|_H}$ we obtain an $x\in D_f\cap H$. This completes proof.
\end{proof}

%
%
%

Two vectors $x, y\in \fqm$ are said to be orthogonal under the quadratic form $f$, denoted by $x\perp y$, if $\ell_f(x,y)=0$. For a $d$-dimensional subspace $H\subseteq \fqm$, its dual space $H^\perp$ under the quadratic form $f$ is defined by
\[ H^{\perp}=\{x\in \fqm \, :\, \ell_f(x,y)=0\, \, {\rm for}\ {\rm all}\,\, y\in H\}. \]
Two subspaces $H$ and $W$ of $\fqm$ are said to be orthogonal, denoted by $H\perp W$, if $x\perp y$ for all $ x\in H$ and $y\in W$.
The self-orthogonal subspace $H^{\perp}_{f|_H}$ of $H$ under $f|_{H}$ is defined as
$$H^{\perp}_{f|_{H}}=\left\{ x\in H\, : \, \ell_f(x,y)=0\,\, {\rm for}\,\, {\rm all}\,\, y\in H \right\}.$$
It is obvious that $H^{\perp}_{f|_{H}}=H^{\perp}\cap H$. Similar to (\ref{eq:rankf}), the rank of $f|_H$ can be represented as follows:
\begin{equation}\label{eq:rankofR}
	{\rm rank}\, f|_{H}=d-{\rm dim}_{\bF_q}\, (H^{\perp}\cap H)+{\rm type}\, f|_{H},
\end{equation}
where ${\rm type}\, f|_{H}={\rm dim}_{\bF_q}\, ({\rm ker}\, \ell_{f|_{H}}/{\rm ker}\, f|_{H})$.

For a subspace $H\subset \bF_{q^m}$, it is known from Proposition~\ref{lem:intersection} that  $|D_f \cap H|$ is equal to the number of solutions for
the equation $f|_H(x)=a$, where $f|_H$ is the restriction of $f$ to $H$. By Lemma~\ref{lem:normalform}, the rank and the standard type of $f|_H$
need to be determined for solving the equation $f|_H(x)=a$. This is different from and more difficult than the case of finite fields of odd characteristic.
The following proposition gives a formula to calculate the intersection of the preimage of a quadratic form and a subspace of $\fqm$, which is
slightly different from Proposition~1 in \cite{Li2021}.

\begin{prop}\label{Proposition1}
Let $f$ be a quadratic form over $\fqm$ and $D_f$ be a set given in (\ref{eq:definingset}). Let $H$ be a $d$-dimensional subspace of $\fqm$ and $R={\rm rank}\, f|_{H}\, (R\leq d)$. Then
\begin{equation*}
	|D_f\cap H|=\begin{cases}
		q^{d-1}\pm q^{d-\frac{R+2}{2}}, &{\rm if}\ R\equiv 0\pmod{2},\\
		q^{d-1}, &{\rm if}\ R\equiv 1\pmod{2}.
	\end{cases}
\end{equation*}
\end{prop}
\begin{proof}
By Proposition~\ref{lem:intersection}, we know that $|D_f\cap H|= |D_{f|_{H}}|$, where
\[|D_{f|_{H}}| = \left\{ x\in H \,:\, f|_H(x) = a,\, a\in {\mathbb F}_q^*  \right\} .\]
It is known that $f|_H(x)$ is a quadratic form over $H$ with rank $R\,(R\leq d)$. By Lemma~\ref{lem:normalform}, if $R$ is even, then the quadratic form
$f|_H(x)$ is equivalent to the standard type I or II, i.e.,
$$x_1x_2+\dots+x_{R-1}x_R \,\, {\rm or}\,\, x_1x_2+\dots+\alpha x_{R-1}^2+x_{R-1}x_R+\alpha x_R^2,$$
and $|D_{f|_{H}}|= q^{d-1}\pm q^{d-\frac{R+2}{2}}$, respectively. If $R$ is odd, then the quadratic form
$f|_H(x)$ is equivalent to the standard type III, i.e.,
	$$ x_1x_2+\dots+x_{R-2}x_{R-1}+x_R^2,$$
and $|D_{f|_{H}}|= q^{d-1}$.
\end{proof}

By a similar proof of Proposition~2 in \cite{Li2021}, we have

\begin{prop}\cite[Proposition 2]{Li2021}\label{pro:orthogonal}
Let $f$ be a non-degenerate quadratic form over $\fqm$. For each $r$ with $0<2r<m$, there exists an $r$-dimensional subspace $H\subseteq \fqm$ such that $H\subseteq H^{\perp}$.
\end{prop}	


\section{The weight hierarchies of linear codes from non-degenerate quadratic forms on ${\mathbb F}_{q^m}$ }\label{Sec:non-degenerate}
Let $f$ be a non-degenerate quadratic form over ${\mathbb F}_{q^m}$, where $q$ is a power of $2$. In this section, by analysis of behavior of the restriction of $f$ to subspaces of $\bF_{q^m}$, we determine the weight hierarchy of the linear code
\begin{equation}\label{eq:defcode}
\C_{D_f} = \left\{  {\bf c}_x = \left( d x \right)_{d\in D_f} \,\, :\,\, x\in {\mathbb F}_{q^m} \right\},
\end{equation}
where $D_f =\left\{ x\in {\mathbb F}_{q^m} : f(x)=a,\, a\in {\mathbb F}_q^* \right\}.$

\begin{theorem}\label{thm:typei}
Let $m$ be an even number. Let $f$ be a non-degenerate quadratic form over ${\mathbb F}_{q^m}$, which is equivalent to Type I. Then the linear code $\C_{D_f}$ defined in (\ref{eq:defcode}) has the following weight hierarchy:
	 \begin{equation*}
	 	d_r\left( \C_{D_f}\right)=\begin{cases}
	 	    q^{m-1}-q^{\frac{m-2}{2}}-q^{m-2}, &{\rm if}\ r=1,\\
	 		q^{m-1}-q^{\frac{m-2}{2}}-q^{m-r-1}-q^{\frac{m-4}{2}}, &{\rm if}\ 2\leq r\leq \frac{m}{2},\\
	 		q^{m-1}-q^{\frac{m-2}{2}}-q^{m-r-1}-q^{m-r-2}, &{\rm if}\ \frac{m}{2}<r<m-1,\\
	 		q^{m-1}-q^{\frac{m-2}{2}}-1, &{\rm if}\ r=m-1,\\
	 		q^{m-1}-q^{\frac{m-2}{2}}, &{\rm if}\ r=m.
	 	\end{cases}
	 \end{equation*}
\end{theorem}

\begin{proof}
Since $f$ is a non-degenerate quadratic form over ${\mathbb F}_{q^m}$, which is equivalent to Type I, by Lemma~\ref{lem:normalform} we have that
$$n= | \left\{ x\in {\mathbb F}_{q^m} \,|\, f(x) =a,\,a\in {\mathbb F}_q^*  \right\}| = q^{m-1}-q^{\frac{m-2}{2}}.$$
By Lemma \ref{lem:GHW}, the weight hierarchy of $\C_{D_f}$ is
$$ d_r(\C_{D_f})=n-{\rm max} \left\{ |D_f \cap H_r^\perp| \, \,: \,\, H_r\in [{\mathbb F}_{q^m},r] \right\}. $$
So, the next task is to determine the maximum value of $|D_f\cap H_r^{\perp}|$ for all $r$-dimension subspaces of $\bF_{q^m}$. By Proposition \ref{Proposition1}, we know that
\begin{equation}\label{eq:objvalue}
| D_f \cap H_r^{\perp}|=
\begin{cases}
	q^{m-r-1}\pm q^{m-r-\frac{R+2}{2}}, & {\rm if}\ R\equiv0\pmod{2},\\
	q^{m-r-1}, & {\rm if}\ R\equiv1\pmod{2},
\end{cases}
\end{equation}
where $R={\rm rank} \ f|_{H_r^{\perp}}$. By the equation (\ref{eq:rankofR}),
\begin{equation}\label{eq:R}
R={\rm rank} \ f|_{H_r^{\perp}}=m-r-{\rm dim}_{\bF_q}(H_{r}\cap H_{r}^{\bot})+ {\rm type}\ f|_{H_r^{\perp}},
\end{equation}
where $R$ is even when ${\rm type}\ f|_{H_r^{\perp}}=0$ and $R$ is odd when ${\rm type}\ f|_{H_r^{\perp}}=1$. From (\ref{eq:objvalue}) we know that the maximum value of $|D_f\cap H_r^{\perp}|$ is $q^{m-r-1}+ q^{m-r-\frac{R+2}{2}}$ with possible minimum value of $R$. 

By Proposition~\ref{lem:intersection}, we know that $|D_f\cap H_r^{\perp}| = |D_{f|_{H_r^{\perp}}}|$. Let $f|_{H_r^{\perp}}$ denote the restriction of $f$ to
the subspace $H_r^\perp$, which is a quadratic form over $H_r^\perp$ at most $m-r$ variables. If $|D_f\cap H_r^{\perp}|$ is equal to $q^{m-r-1}+ q^{m-r-\frac{R+2}{2}}$, by Lemma~\ref{lem:normalform},
then $f|_{H_r^{\perp}}(x)$ must be equivalent to a quadratic form of Type II.

Then, we will determine the minimum value of $R$ when $f|_{H_r^{\perp}}$ is equivalent to a quadratic form over $H_r^\perp$ of Type II.
By Proposition~\ref{pro:orthogonal} and (\ref{eq:R}), we know that $R$ reaches its minimum value when  $H_{r}\subseteq H_{r}^{\perp}$ or $H_{r}^{\perp}\subseteq H_{r}$ i.e.,
$r\leq m-r$ or $m-r\leq r$. These two cases are discussed below.

 {\bf (1)} $1\leq r\leq \frac{m}{2}$. It is known that $ 0\leq {\rm dim}_{\bF_q}\,(H_{r}\cap H_{r}^{\perp} )\leq r$. From (\ref{eq:R}) we have
\begin{equation}\label{neq:R}
	m-2r+{\rm type}\ f|_{H_r^{\perp}}\leq R\leq m-r+ {\rm type}\ f|_{H_r^{\perp}}.
\end{equation}
By (\ref{neq:R}), the minimum value of $R$ is $m-2r$ when type\, $f|_{H_r^{\perp}}=0$. Next, we will show that $f|_{H_r^{\perp}}$ can only be equivalent to Type I when $R=m-2r$.

By Proposition~\ref{pro:orthogonal}, we can construct an $r$-dimensional subspace $H_r$ of $\fqm$ and its dual space as follows:
$$H_{r}=\langle \beta_1,\beta_2,\dots,\beta_r\rangle,\quad H_{r}^\perp =H_{m-r}=\langle\alpha_1,\alpha_2,\dots,\alpha_{m-2r},\beta_1,\beta_2,\dots,\beta_r\rangle.$$
Since ${\rm dim}_{\bF_q}(H_r\cap H_r^\perp) =r$ we have rank\ $f|_{H_r^\perp}= m-2r$ by (\ref{eq:R}). From $H_r$ and its dual space $H_{m-r}$ we set
$$ H_{m-2r}=\langle \alpha_1,\alpha_2,\dots,\alpha_{m-2r}\rangle,\quad H_{m-2r}^{\perp}=\langle\beta_1,\beta_2,\dots,\beta_r,\beta_{r+1},\dots,\beta_{2r}\rangle.$$
It is clear that $\B_1=\{\alpha_1,\alpha_2,\dots,\alpha_{m-2r}, \beta_1,\beta_2,\dots,\beta_{2r}\}$ is a basis of $\fqm$ over $\bF_q$. Under this basis, the matrix of the quadratic form $f$ is
$$ M_f(\B_1)=\left(\begin{array}{cccc}
		f(\alpha_1) & \ell_f(\alpha_1,\alpha_2) & \cdots & \ell_f(\alpha_1,\beta_{2r})\\
		 & f(\alpha_2) &  & \vdots \\
		 & & \ddots & \vdots \\
		 & & & f(\beta_{2r})
	\end{array}\right).$$
Since rank $f|_{H_r^\perp}= m-2r$ is even, we know that type\ $f|_{H_r^\perp}=0$. This implies that ker\ $\ell_{f|_{H_r^\perp}}$=ker\ $f|_{H_r^\perp}$.
So,
$$ f(\beta_i)=\ell_f(\beta_i,\beta_i)=0, \,\,  1\leq i\leq r .$$
Therefore, $M_f(\B_1)$ has the following form,
$$ M_f(\B_1)=\left( \begin{array}{ccc}
		M_{m-2r} & 0 & 0 \\
		0 & 0_{r}  & K_r \\
		0 & 0  & M_r
	\end{array}
\right),$$
where $M_{m-2r}$ is an upper triangular matrix with order $m-2r$. Since $f$ is non-degenerate, $M_f(\B_1)+M_f^T(\B_1)$ is nonsingular. This implies that $K_r$ is a nonsingular matrix of order $r$.
It is easy to verify that
\begin{equation}\label{matrixf}
G^{T} M_f(\B_1)G=\left(
	\begin{array}{ccc}
		M_{m-2r} & 0 & 0 \\
		0 & 0_r  & I_r \\
		0 & 0  & 0_r
	\end{array}
	\right), \,\,{\rm where}\,\,
G=\left(
	\begin{array}{ccc}
		I_{m-2r} & 0 & 0 \\
		0 & I_r & (K_r^{-1})^{T}M_r^{T}K_r^{-1}\\
		0 & 0 & K_r^{-1}
	\end{array}
	\right).
\end{equation}
According to the representations of $H_r^\perp$ and $H_{m-2r}$, we know that the restrictions of $f$ to $H_r^\perp$ and $H_{m-2r}$ have the same canonical representation.
Assume that $f|_{H_r^{\perp}}$ is equivalent to a quadratic form on $H_r^{\perp}$ of Type II, and then $f|_{H_{m-2r}}$ is also equivalent to a quadratic form on $H_{m-2r}$ of Type II. Hence, from (\ref{matrixf}) we
derive that $f$ is equivalent to a quadratic form on ${\mathbb F}_{q^m}$ of Type II. This is a contradiction to that $f$ is equivalent to Type I.

Next, we construct a subspace $H_r^{\perp}$ such that $f|_{H_r^{\perp}}$ is equivalent to a quadratic form on $H_r^\perp$ of Type II and $R=$ rank $f|_{H_r^\perp} =m-2r+2$. By Proposition~\ref{pro:orthogonal}, there exist subspaces $H_{r-2}, H_{r-2}^\perp \subset \bF_{q^m}$ as follows:
$$H_{r-2}=\langle\beta_1,\dots,\beta_{r-2}\rangle, \quad H_{r-2}^{\perp}= H_{m-r+2}=\langle\beta_1,\beta_2,\dots,\beta_{r-2},\alpha_1,\alpha_2,\dots,\alpha_{m-2r+4}\rangle.$$
From $H_{r-2}$ we construct two subspaces $H_r$ and $H_r^\perp$ as follows:
$$H_{r}=\langle\beta_1,\dots,\beta_{r-2},\xi_1,\xi_2\rangle, \quad H_r^\perp = H_{m-r}=\langle\delta_1,\delta_2,\dots,\delta_{m-2r+2},\beta_1,\dots,\beta_{r-2}\rangle,$$
where $\{\xi_1,\xi_2\}\in H_{r-2}^{\perp}$.

Assume that $f|_{H_r^\perp}$ is equivalent to a quadratic form over $H_r^\perp$ of Type II. Since ${\rm dim}_{\bF_q}\,(H_r \cap H_r^\perp)=r-2$, from (\ref{eq:R}) we have
$R={\rm rank} f|_{H_{r}^{\perp}}=m-2r+2$. Next, we show that there exists a basis of $\bF_{q^m}$ over $\bF_q$ such that $f$ is equivalent to a quadratic form over
$\bF_{q^m}$ of Type I under this basis. Set
$$ H_{m-2r+2}=\langle\delta_1,\delta_2,\dots,\delta_{m-2r+2}\rangle,\quad H_{m-2r+2}^{\perp}=\langle\beta_1,\beta_2,\dots,\beta_{r-2},\gamma_1,\gamma_2,\dots,\gamma_{r-2},\xi_1,\xi_2\rangle.$$
Under the basis $\B_2=\{\delta_1,\delta_2,\dots,\delta_{m-2r+2},\beta_1,\dots,\beta_{r-2},\gamma_1,\dots,\gamma_{r-2},\xi_1,\xi_2\}$ of ${\mathbb F}_{q^m}$ over $\bF_q$, the matrix of $f$ has the following form,
$$M_f(\B_2)=\left(
	\begin{array}{cccc}	
	M_{m-2r+2} & 0 & 0 & 0\\
	 0 & 0_{r-2} & K_{r-2} & 0\\
	 0 & 0 & M_{r-2} & N\\
	 0 & 0 & 0 & M_2
	\end{array}
	\right),$$	
where $M_{m-2r+2}$ and $M_2$ are upper triangular matrices with order $m-2r+2$ and $2$, respectively. Since $f$ is non-degenerate, $M_f(\B_2)+M_f^T(\B_2)$ is non-singular. This implies that $K_{r-2}$ is a non-singular matrix of order $r-2$.
It is verified that
$$ M_f^\prime(\B_2) = G^TM_f(\B_2)G=\left(
	\begin{array}{cccc}
		M_{m-2r+2} & 0 & 0 & 0 \\
		 0 & 0_{r-2} & I_{r-2} & 0\\
		 0 & 0 & 0 & (K_{r-2}^{-1})^T N\\
		 0 & 0 & N^TK_{r-2}^{-1} & M_2
	\end{array}
	\right),$$
where $M_2$ is congruent to a matrix of the form
$$ \left(\begin{array}{cc}
		\alpha & 1\\
		0 & \alpha
	\end{array}\right), \,\, {\rm and } \,\, G=\left(
	\begin{array}{cccc}
		I_{m-2r+2} & 0 & 0 & 0\\
		0 & I_{r-2} & (K_{r-2}^{-1})^T M_{r-2}^T K_{r-2}^{-1} & (K_{r-2}^{-1})^T N\\
		0 & 0 & K_{r-2}^{-1} & 0\\
		0 & 0 & 0 & I_2
	\end{array}\right).
$$
Moreover, it is easy to show that ${\bf x} M_f^\prime(\B_2) {\bf x}^T = {\bf x} \bar{M}_f(\B_2) {\bf x}^T$, where ${\bf x}=(x_1, x_2, \cdots, x_m)$ and
$$ \bar{M}_f(\B_2)=\left(
	\begin{array}{cccc}
		M_{m-2r+2} & 0 & 0 & 0 \\
		 0 & 0_{r-2} & I_{r-2} & 0\\
		 0 & 0 & 0_{r-2} & 0\\
		 0 & 0 & 0 & M_2
	\end{array}
	\right),
$$
that is to say, the matrix of $f$ under the basis $\B_2$ is $\bar{M}_f(\B_2)$.

According to the representations of $H_r^\perp$ and $H_{m-2r+2}$, we know that the restrictions of $f$ to $H_r^\perp$ and $H_{m-2r+2}$ have the same canonical representation.
Since $f|_{H_r^{\perp}}(x)$ is equivalent to a quadratic form on $H_r^{\perp}$ of Type II, $f|_{H_{m-2r+2}}(x)$ is also equivalent to a quadratic form on $H_{m-2r+2}$ of Type II.
So, $M_{m-2r+2}$ is a the matrix of Type II.  By Lemma 11.17 in~\cite{Wan2003}, we derive that $f$ is equivalent to a quadratic form on ${\mathbb F}_{q^m}$ of Type I.
Hence, $R$ can reach $m-2r+2$ when $f|_{H_r^{\perp}}$ is equivalent to the quadratic form on $H_r^\perp$ of Type II. In this case, from (\ref{eq:R}) we know that
${\rm dim}_{\bF_q}\,(H_{r}\cap H_{r}^{\perp})=m-r-R=r-2$. Therefore, the cases $r=1$ and $2\leq r\leq m/2$ need to be discussed separately.

{\bf Case 1:} $r=1$. From (\ref{neq:R}), the possible minimum values of $R$ are $m-2$ and $m-1$. According to analysis above, we know that $f|_{H_1^{\perp}}$ is equivalent to a quadratic form on $H_1^\perp$ of Type I when $R=m-2$, and $f|_{H_1^{\perp}}$ is equivalent to the quadratic form on $H_1^\perp$ of Type III when $R=m-1$.
By~(\ref{eq:objvalue}), we know that $|D_f\cap H_1^{\perp}|$ is maximized when $R=m-1$. Next, we show that there exists a basis under which
$f$ is equivalent to Type I and its restriction to $H_1^\perp$ is equivalent to Type III.

By Proposition \ref{pro:orthogonal}, there exists a $1$-dimensional subspace $H_1$ of ${\mathbb F}_{q^m}$ and its dual space as follows:
		$$H_1=\langle\alpha_{m-1}\rangle,\quad H_1^{\perp}=H_{m-1}=\langle\alpha_1,\alpha_2,\dots,\alpha_{m-1}\rangle.$$
So, $R={\rm rank}\, f|_{H_{1}^{\perp}}=m-1-({\rm dim}\,(H_{1}\cap H_{1}^{\perp})-{\rm type}\, f|_{H_{1}^{\perp}})=m-1$. From $1$-dimensional space $H_1$ and its dual space we construct a $2$-dimensional space $H_2$ and its dual space as follows:
$$ H_2=\langle\alpha_{m-1},\gamma\rangle, \,\,  H_2^\perp = H_{m-2}=\langle\alpha_1,\alpha_2,\dots,\alpha_{m-2}\rangle.$$
It is clear that $f|_{H_2^\perp}$ is a non-degenerate quadratic form on $H_2^\perp$. Then, ${\rm dim}_{\bF_q} (H_2\cap H_2^\perp) =0$. So,
$\B_3=\{\alpha_1,\alpha_2,\dots,\alpha_{m-2},\alpha_{m-1},\gamma\}$ is a basis of ${\mathbb F}_{q^m}$ over ${\mathbb F}_q$. Since $R=m-1$ is odd, the
quadratic form $f|_{H_1^{\perp}}(x)$ has the canonical form of Type III and its corresponding matrix has the following form:
$$ M_{m-1}=\left(
	\begin{array}{cc}
		M_{m-2} & 0 \\
		0 & 1
	\end{array}\right),$$
where $M_{m-2}$ is a matrix of a quadratic form of Type I. For the sake of convenience, we can assume that $f(\gamma)=\alpha=\rho^2 + \rho$, where $\rho\in\bF_q^*$ and $\ell_f(\alpha_{m-1},\gamma)=1$. Under the basis $\B_3$, the matrix of quadratic form $f$ is
$$ M_f(\B_3)=\left(
	\begin{array}{ccc}
		M_{m-2} & 0 & 0\\
		0 & 1 & 1\\
		0 & 0 & \alpha
	\end{array}\right).$$
It is easy to verify that
$$M_f^\prime(\B_3)=G^T M_f(\B_3)G=\left(
	\begin{array}{ccc}
		M_{m-2} & 0 & 0\\
		0 & 0 & 1+\alpha+\rho^{-2}\alpha^2\\
		0 & \alpha+\rho^{-2}\alpha^2 & 0
	\end{array}\right), $$
where
$$G=\left(
	\begin{array}{ccc}
		I_{m-2} & 0 & 0\\
		0 & \rho & \rho^{-1}\alpha\\
		0 & 1 & \rho^{-2}\alpha+\rho^{-1}
	\end{array}\right). $$
Moreover, it is clear that ${\bf x} M_f^\prime(\B_3) {\bf x}^T = {\bf x} \bar{M}_f(\B_3) {\bf x}^T$, where ${\bf x}=(x_1, x_2, \cdots, x_m)$ and
$$ \bar{M}_f(\B_3)=\left(
	\begin{array}{ccc}
		M_{m-2} & 0 & 0\\
		0 & 0 & 1\\
		0 & 0 & 0
	\end{array}\right).$$
We have shown that $f$ is equivalent to a quadratic form on ${\mathbb F}_{q^m}$ of Type I, and its restriction to $H_1^\perp$ is equivalent to a quadratic form of Type III.
In this case, $|D_f\cap H_1^{\perp}|= q^{m-r-1}$ and
$$ d_1(\C_{D_f})=n- {\rm max} \left\{ |D_f\cap H_1^{\perp}| \,:\, H_1\in [\fqm, 1] \right\} = q^{m-1}-q^{\frac{m-2}{2}}-q^{m-2}.$$

{\bf Case 2:}\, $2\leq r\leq \frac{m}{2}$. It is shown above that $|D_f \cap H_r^{\perp}|$ is maximized when the quadratic form $f|_{H_r^{\perp}}$ is equivalent to Type II and $R=m-2r+2$. So, by Lemma~\ref{lem:normalform}, the number of solutions for the equation $f|_{H_r^{\perp}}(x)=a$ is $q^{m-r-1}+q^{m-r-\frac{R+2}{2}}=q^{m-r-1}+q^{\frac{m-4}{2}}$. Hence,
$$ d_r(\C_{D_f})=n-{\rm max}\left\{ |D_f \cap H_r^\perp| \, :\, H_r\in [{\mathbb F}_{q^m},r] \right\}=q^{m-1}-q^{\frac{m-2}{2}}-q^{m-r-1}-q^{\frac{m-4}{2}}. $$

{\bf (2)} $\frac{m}{2}<r<m$. In this case, $0\leq {\rm dim}_{\bF_q}\,(H_{r}\cap H_{r}^{\perp})\leq m-r$. From (\ref{eq:R}) we have
\begin{equation}\label{neq:R2}
	{\rm type}\, f|_{H_r^{\perp}}\leq R\leq m-r+{\rm type}\, f|_{H_r^{\perp}}.
\end{equation}
From (\ref{eq:objvalue}) we know that $|D_f \cap H_r^\perp|$ is maximized when $f|_{H_r^\perp}$ is equivalent to Type II and $R$ is the smallest possible even number.

From (\ref{neq:R2}), the least even number that $R$ can take is 0. In this case, type\ $f|_{H_r^{\perp}}=0$ and rank $f|_{H_r^{\perp}}=0$. So,
the number of solutions of $f|_{H_r^{\perp}}(x)=a$ for $a\in \bF_q^*$ is $0$, i.e., $|D_f \cap H_r^\perp| = 0$. The second-to-last smallest even number desirable for $R$
is $2$. In this case, from (\ref{eq:R}) we have ${\rm dim}_{\bF_q}\,(H_{r}\cap H_{r}^{\perp})=m-2-r$. When $r\leq m-2$, by a discussion similar to the case $1 \leq r \leq \frac{m}{2}$ above,
we can construct a subspace $H_r$ such that the quadratic form $f|_{H_r^{\perp}}(x)$ on $H_r^\perp$ is equivalent to Type II. The following cases are discussed.

{\bf Case 1:} $\frac{m}{2}<r\leq m-2$. Since $R=2$, the maximum value of $|D_f\cap H_r^\perp|$ is $q^{m-r-1}+q^{m-r-2}$ for any $H_r\in [\fqm, r]$. So,
$$ d_r(\C_{D_f})=n- {\rm max} \left\{ |D_f\cap H_r^\perp|\, :\, H_r\in [{\mathbb F}_{q^m}, r]\right\} =q^{m-1}-q^{\frac{m-2}{2}}-q^{m-r-1}-q^{m-r-2}.$$

{\bf Case 2:} $r=m-1$. Let $H_{m-1}$ be a subspace of $\fqm$ with dimension $m-1$ and its dual space $H_{m-1}^\perp$ has the dimension $1$.
From (\ref{eq:rankofR}), we have
\begin{equation}\label{eqrm-1}
 R={\rm rank}\,f|_{H_{m-1}^{\perp}}=1-{\rm dim}_{\bF_q}\,(H_{m-1}^{\perp}\cap H_{m-1})+{\rm type}\,f|_{H_{m-1}^{\perp}}.
\end{equation}
It is known that $H_{m-1}^{\perp}\cap H_{m-1}$ is the self-orthogonal subspace of $H_{m-1}$ under $f|_{H_{m-1}}$. From (\ref{eq:rankf}) we have that
$R=0$ when ${\rm type}\,f|_{H_{m-1}^{\perp}}=0$, and  $R=1$ when ${\rm type}\,f|_{H_{m-1}^{\perp}}=1$.
If $R=0$, then $|D_f\cap H_{m-1}^{\perp}|= |D_{f_{H_{m-1}^{\perp}}}| =0$ for $a\in \bF_{q^m}^*$. When $R=1$, from (\ref{eq:objvalue}) we have  $|D_f\cap H_{m-1}^\perp|=1$.
Next, we show that there exist subspaces $H_{m-1}$ and $H_{m-1}^\perp$ of $\bF_{q^m}$ such that $R={\rm rank}\,f|_{H_{m-1}^{\perp}}=1$ and $f|_{H_{m-1}^\perp}$ is equivalent to Type III, where $f$ is a quadratic form over $\bF_{q^m}$, which is equivalent to Type I. Let $M_f$ and $M_{f|_{H_{m-1}^{\perp}}}$ be the matrices of the canonical forms of $f$ and $f|_{H_{m-1}^\perp}$, respectively. From (\ref{eq:rankf}) and (\ref{eqrm-1}) we know that
rank $(M_{f|_{H_{m-1}^{\perp}}} + M_{f|_{H_{m-1}^{\perp}}}^{T})=0$. So, any vector in $\bF_{q^m}$ is self-orthogonal under $f|_{H_{m-1}^\perp}$.  Let $\beta=(b, b,0,\dots,0)$ be a self-orthogonal vector under $f|_{H_{m-1}^\perp}$ for $b\in \bF_q^*$, and
\begin{equation*}
H_{m-1}=\langle\beta,\alpha_1,\dots,\alpha_{m-2}\rangle, \,\,\, H_{m-1}^{\perp}=\langle \beta\rangle.
\end{equation*}
It is easy to verify that $\beta M_f\beta^T = b^2$. This shows that $f|_{H_{m-1}^\perp}$ is equivalent to Type III. So,
 $$ d_{m-1}(\C_{D_f})=n-{\rm max} \left\{ |D_f\cap H_{m-1}^\perp| \, :\, H_{m-1} \in [\fqm, m-1] \right\} =q^{m-1}-q^{\frac{m-2}{2}}-1.$$

{\bf Case 3:} $r=m$. From (\ref{eq:R}) we hvae $R={\rm rank}\,f|_{H_m^{\perp}}=0$. So, $|D_f\cap H_m^\perp|=0$ and
$$ d_m(\C_{D_f})=q^{m-1}-q^{\frac{m-2}{2}}.$$
\end{proof}

\begin{example}
Let $w$ be a primitive element of ${\mathbb F}_{2^4}$ and $f(x)={\rm Tr}_1^4( wx^3)$ be a quadratic form on ${\mathbb F}_{2^4}$, where ${\rm Tr}_1^4(\cdot)$ is a trace function from ${\mathbb F}_{2^4}$ to $\bF_2$. Let $\C_{D_f}$ be a linear code as in (\ref{eq:defcode}), where $D_f=\{ x\in \bF_{2^4} \, |\, f(x)=1 \}$.
By the help of Magma, we obtain the weight hierarchy of $\C_{D_f}$ as follows: $d_1=2, d_2=3, d_3=5, d_4=6$.  This result is consistent with Theorem~\ref{thm:typei}.
\end{example}

By a proof similar to Theorem~\ref{thm:typei}, we get the following theorem.

\begin{theorem}\label{thm:typeii}
Let $m$ be an even number. Let $f$ be a non-degenerate quadratic form over ${\mathbb F}_{q^m}$, which is equivalent to Type  \uppercase\expandafter{\romannumeral2}. Then the linear codes $\C_{D_f}$ defined in (\ref{eq:defcode}) has the  following weight hierarchy:
\begin{equation*}
d_r(\C_{D_f})=\begin{cases}
			q^{m-1}-q^{m-r-1}, &{\rm if}\ 1\leq r\leq \frac{m}{2}-1,\\
			q^{m-1}+q^{\frac{m-2}{2}}-q^{m-r-1}-q^{m-r-2}, &{\rm if}\ \frac{m}{2}\leq r<m-1,\\
			q^{m-1}+q^{\frac{m-2}{2}}-1, &{\rm if}\ r=m-1,\\
			q^{m-1}+q^{\frac{m-2}{2}}, &{\rm if}\ r=m.
		\end{cases}
	\end{equation*}
\end{theorem}

\begin{example}
Let $w$ be a primitive element of ${\mathbb F}_{2^6}$ and $f(x)={\rm Tr}_1^6( wx^3)$ be a quadratic form on ${\mathbb F}_{2^6}$, where ${\rm Tr}_1^6(\cdot)$ is a trace function from ${\mathbb F}_{2^6}$ to $\bF_2$. Let $\C_{D_f}$ be a linear code as in (\ref{eq:defcode}), where $D_f=\{ x\in \bF_{2^6} \, |\, f(x)=1 \}$.
By the help of Magma, we obtain the weight hierarchy of $\C_{D_f}$ as follows: $d_1=16, d_2=24, d_3=30, d_4=33, d_5=35, d_6=36$. This result is consistent with Theorem~\ref{thm:typeii}.	
\end{example}

By a discussion similar to Theorem~\ref{thm:typei}, we obtain the following theorem.

\begin{theorem}\label{thm:typeiii}
Let $m$ be an odd number. Let $f$ be a non-degenerate quadratic form over ${\mathbb F}_{q^m}$, which is equivalent to Type III. Then the linear code $\C_{D_f}$ defined in (\ref{eq:defcode}) has the following weight hierarchy:
\begin{equation*}
		d_r(\C_{D_f})=\begin{cases}
			q^{m-1}-q^{m-r-1}-q^{\frac{m-3}{2}}, &{\rm if}\ 1\leq r\leq \frac{m-1}{2},\\
			q^{m-1}-q^{m-r-1}-q^{m-r-2}, &{\rm if}\ \frac{m+1}{2}\leq r<m-1,\\
			q^{m-1}-1, &{\rm if}\ r=m-1,\\
			q^{m-1}, &{\rm if}\ r=m.
		\end{cases}
	\end{equation*}
\end{theorem}

\begin{example}
Let $w$ be a primitive element of ${\mathbb F}_{2^5}$ and $f(x)={\rm Tr}_1^5( wx^3)$ be a quadratic form on ${\mathbb F}_{2^5}$, where ${\rm Tr}_1^5(\cdot)$ is a trace function from ${\mathbb F}_{2^5}$ to $\bF_2$. Let $\C_{D_f}$ be a linear code as in (\ref{eq:defcode}), where $D_f=\{ x\in \bF_{2^5} \, |\, f(x)=1 \}$.
By the help of Magma, we obtain the weight hierarchy of $\C_{D_f}$ as follows: $d_1=6, d_2=10, d_3=13, d_4=15, d_5=16$. This result is consistent with Theorem~\ref{thm:typeiii}.	
\end{example}

\section{The weight hierarchies of linear codes from degenerate quadratic forms on ${\mathbb F}_{q^m}$}\label{Sec:degenerate}

In this section, we discuss the weight hierarchy of the linear code $\C_{D_f}$ defined in (\ref{eq:defcode}) for $f(x)$ being a degenerate quadratic form over ${\mathbb F}_{q^m}$.
From now on, let $f$ be a degenerate quadratic form over ${\mathbb F}_{q^m}$ and $\bar{{\mathbb F}}_{q^m}={\mathbb F}_{q^m}/{\rm ker}\, f$. It is easy to see that
the quadratic form $f$ induces a non-degenerate quadratic form $\bar{f}$ over $\bar{\bF}_{q^m}$ as follows:
\begin{equation*}
    	\begin{aligned}
    		\bar{f} \,:\, \bar{\bF}_{q^m}\, &\longrightarrow {\bF_q} \\
    		  \bar{x} \,\,\,  &\longmapsto f(x).
    	\end{aligned}
    \end{equation*}
Let $\varphi \,:\, x \mapsto \bar{x}$ be a canonical map from $\bF_{q^m}$ to $\bar{\bF}_{q^m}$. For a subspace $H\subset \bF_{q^m}$, $\varphi(H)= H/(H\cap {\rm ker}\,f)$ is
a subspace of $\bar{\bF}_{q^m}$, and denote it by $\bar{H}$. Let $\bar{f}|_{\bar{H}}$ denote the restriction of the quadratic form of $\bar{f}$ to $\bar{H}$, and
$\bar{R}$ denote the rank of $\bar{f}|_{\bar{H}}$. The dual space of $\bar{H}$ under the quadratic form $\bar{f}$ is defined as
$$\bar{H}^{\perp}=\{\bar{x}\in \bar{{\mathbb F}}_{q^m}: \ell_{\bar{f}}(\bar{x},\bar{y})=0\,\, {\rm for\,\, all}\,\, \bar{y}\in \bar{H}\},$$
and the self-orthogonal subspace of $\bar{H}$ under $\bar{f}|_{\bar{H}}$ is defined as
$${\bar{H}_{f|_{\bar{H}}}^{\perp}}=\{\bar{x}\in \bar{H}: \ell_{\bar{f}}(\bar{x},\bar{y})=0\,\, {\rm for\,\, all}\,\, \bar{y}\in \bar{H}\}.$$

For a degenerate quadratic form $f$, to calculate the $r$-th generalized Hamming weight of $\C_{D_f}$ in~(\ref{eq:defcode}), we need to find the maximum value of $|D_f\cap H_r^{\perp}|$ for all $r$-dimensional
subspace $H_r$ of ${\mathbb F}_{q^m}$, which is determined by the rank and type of $f|_{H_r^{\perp}}$ by Proposition~\ref{Proposition1}. Since $\bar{f}(\bar{x})=f(x)$ for all $x\in {\mathbb F}_{q^m}$, we know
that $f|_{H}$ and $\bar{f}|_{\bar{H}}$ have the same rank and standard type. However, Since $\bar{f}$ is a non-degenerate quadratic form on $\bar{\bF}_{q^m}$, the rank and type of $\bar{f}|_{\bar{H}}$ can be determined as we did in Theorem~\ref{thm:typei}. Hence, by a similar way to the last section, we can determine the weight hierarchy of the linear code $\C_{D_f}$ for $f$ being a degenerate
quadratic form over $\bF_{q^m}$.


\begin{theorem}\label{thm:iv}
Let $m$ be a positive integer and $f$ be a degenerate quadratic form over ${\mathbb F}_{q^m}$ with $rank\ f=2s\,(2s<m)$, which is equivalent to Type I.  Then the linear code $\C_{D_f}$ defined in (\ref{eq:defcode}) has the following weight hierarchy:
	\begin{equation*}
	 	d_r(\C_{D_f})=\begin{cases}
	 	    q^{m-1}-q^{m-s-1}-q^{m-2}, &{\rm if}\ r=1,\\
	 		q^{m-1}-q^{m-s-1}-q^{m-r-1}-q^{m-s-2}, &{\rm if}\ 2\leq r\leq s,\\
	 		q^{m-1}-q^{m-s-1}-q^{m-r-1}-q^{m-r-2}  , &{\rm if}\ s<r<m-1,\\
	 		q^{m-1}-q^{m-s-1}-1, &{\rm if}\ r=m-1,\\
	 		q^{m-1}-q^{m-s-1}, &{\rm if}\ r=m.
	 	\end{cases}
	 \end{equation*}
\end{theorem}

\begin{proof}
Since $f$ is a degenerate quadratic form over ${\mathbb F}_{q^m}$ with ${\rm rank}\,f=2s$, which is equivalent to Type \uppercase\expandafter{\romannumeral1}, by Lemma \ref{lem:normalform} we know that
$$n=|\,\{x\in {\mathbb F}_{q^m}\,|\,f(x)=a,\,a\in {\mathbb F}_q^*\}\,|=q^{m-1}-q^{m-s-1}.$$
By Lemma \ref{lem:GHW}, the $r$-th generalized Hamming weight of $\C_{D_f}$ is equal to
$$ d_r(\C_{D_f})=n-{\rm max} \left\{ |D_f \cap H_r^\perp| \, \,: \,\, H_r\in [{\mathbb F}_{q^m},r] \right\}, $$
where $[{\mathbb F}_{q^m},r]$ is the set of all $r$-dimensional subspaces of $\bF_{q^m}$. By Proposition \ref{Proposition1}, we know that
\begin{equation}\label{eq:objvalue1}
| D_f \cap H_r^{\perp}|=
\begin{cases}
	q^{m-r-1}\pm q^{m-r-\frac{R+2}{2}}, & {\rm if}\ R\equiv0\pmod{2},\\
	q^{m-r-1}, & {\rm if}\ R\equiv1\pmod{2},
\end{cases}
\end{equation}
where $R={\rm rank} \ f|_{H_r^{\perp}}$. From (\ref{eq:objvalue1}), the possible maximum value of $|D_f\cap H_r^{\perp}|$ is $q^{m-r-1}+q^{m-r-\frac{R+2}{2}}$ when $R$ is the smallest even number possible.
In this case, by Lemma \ref{lem:normalform}, $f|_{H_r^{\perp}}(x)$ is equivalent to a quadratic form over $H_r^\perp$ of Type II.

Next, we determine the minimum value of $R$ when $f|_{H_r^{\perp}}(x)$ is equivalent to a quadratic form over $H_r^\perp$ of Type II.
Let $\bar{H}_r$ denote the image of $H_r$ under the canonical map $\varphi$, i.e., $\bar{H}_r = H_r/H_r\cap {\rm ker}\, f$. Let $\bar{R}$ denote the rank of $\bar{f}|_{\bar{H}}$.
From (\ref{eq:rankofR}) we have
\begin{equation}\label{eq:rankofR4}
\bar{R}={\rm rank}\,\bar{f}|_{\bar{H}_r^{\perp}}={\rm dim}_{\bF_q}\,(\bar{H}_r^{\perp})-{\rm dim}_{\bF_q}\,(\bar{H}_{r}\cap \bar{H}_{r}^{\bot})+{\rm type}\, \bar{f}|_{\bar{H}_r^{\perp}}
\end{equation}
and $\bar{R} = R$. Since $\bar{f}$ is non-degenerate over $\bar{{\mathbb F}}_{q^m}$, by Proposition~\ref{pro:orthogonal} and (\ref{eq:rankofR4}), $\bar{R}$ reaches its minimum value when $\bar{H}_r\subseteq \bar{H}_r^{\perp}$ or $\bar{H}_r^{\perp}\subseteq \bar{H}_r$, i.e., $r\leq 2s-r$ or $2s-r\leq r$. These two cases are discussed below.

{\bf (1)} $1\leq r\leq s$. In this case, for a subspace $H_{r}\subset {\mathbb F}_{q^m}$ we have
\[\begin{split}
{\rm dim}_{\bF_q}\,(\bar{H}_{r}^{\perp})={\rm dim}_{\bF_q}\,(H_{r}^{\perp})-{\rm dim}_{\bF_q}\,(H_{r}^{\perp}\cap {\rm ker}\, f)\geq m-r-(m-2s)=2s-r, &\,\, {\rm and} \\
0\leq {\rm dim}_{\bF_q}\,(\bar{H}_{r}\cap \bar{H}_{r}^{\bot})\leq {\rm dim}_{\bF_q}\,(\bar{H}_{r})=2s-{\rm dim}_{\bF_q}\,(\bar{H}_{r}^{\perp})\leq 2s-(2s-r)=r.&
\end{split}\]
So, from (\ref{eq:rankofR4}), we know that
\begin{equation}\label{neq:barR}
\bar{R}\geq 2s-2r+{\rm type}\, \bar{f}|_{\bar{H}_r^{\perp}}.
\end{equation}
Since $\bar{f}$ is non-degenerate over $\bar{\bF}_{q^m}$, by a discussion similar to Theorem~\ref{thm:typei}, we know that the minimum desirable value for $\bar{R}$ is $2s-2r+2$, rather than $2s-2r$
when $\bar{f}|_{\bar{H}_r^{\perp}}(\bar{x})$ is equivalent to the quadratic form over $\bar{H}_r^\perp$ of Type II. Next, we construct a subspace $\bar{H}_r^{\perp}\subset \bar{\bF}_{q^m}$
such that $\bar{f}|_{\bar{H}_r^{\perp}}$ is equivalent to the quadratic form over $\bar{H}_r^\perp$ of Type II and $\bar{R}=2s-2r+2$.

Since ${\rm dim}_{\bF_q}(\bar{\bF}_{q^m}) =2s$ and $1\leq r\leq s$, by Proposition \ref{pro:orthogonal}, there exist subspaces $\bar{H}_{r-2}, \bar{H}_{r-2}^{\perp} \subset \bar{\bF}_{q^m}$ as follows:
$$\bar{H}_{r-2}=\langle\beta_1,\dots,\beta_{r-2}\rangle, \quad \bar{H}_{r-2}^{\perp}= \bar{H}_{2s-r+2}=\langle \alpha_1,\alpha_2,\dots,\alpha_{2s-2r+4},\beta_1,\beta_2,\dots,\beta_{r-2}\rangle.$$
From $\bar{H}_{r-2}$ and $\bar{H}_{r-2}^{\perp}$ we can construct two subspaces $\bar{H}_r$ and $\bar{H}_r^{\perp}$ of $\bar{\bF}_{q^m}$ as follows:
$$\bar{H}_r=\langle\beta_1,\beta_2,\dots,\beta_{r-2},\xi_1,\xi_2\rangle,\quad \bar{H}^{\perp}_{r}=\bar{H}_{2s-r}=\langle \delta_1,\delta_2,\dots,\delta_{2s-2r+2},\beta_1,\beta_2,\dots,\beta_{r-2}\rangle, $$
where $\{\xi_1,\xi_2\}\in \bar{H}_{r-2}^{\perp}$.

Assume that $\bar{f}|_{\bar{H}_r^\perp}$ is equivalent to a quadratic form over $\bar{H}_r^\perp$ of Type II. Since ${\rm dim}_{{\mathbb F}_q}(\bar{H}_{r}\cap \bar{H}_{r}^{\perp})=r-2$, from (\ref{eq:rankofR4}) we have
$\bar{R}={\rm rank}\,\bar{f}|_{\bar{H}_r^{\perp}}=2s-2r+2$. Next, we show that there exists a basis of $\bF_{q^m}$ over $\bF_q$ such that $f$ is equivalent to a quadratic form over $\bF_{q^m}$ of Type I under this basis.
Set
$$\bar{H}_{2s-2r+2}=\langle \delta_1,\delta_2,\dots,\delta_{2s-2r+2}\rangle,\quad \bar{H}_{2s-2r+2}^{\perp}=\langle\beta_1,\beta_2,\dots,\beta_{r-2},\gamma_1,\gamma_2,\dots,\gamma_{r-2},\xi_1,\xi_2\rangle,$$
and
$${\rm ker}\,f=\langle\theta_1,\theta_2,\dots,\theta_{m-2s}\rangle.$$
Under the basis $\B=\{\delta_1,\dots,\delta_{2s-2r+2},\beta_1,\dots,\beta_{r-2},\gamma_1,\dots,\gamma_{r-2},\xi_1,\xi_2,\theta_1,\dots,\theta_{m-2s}\}$ of ${\mathbb F}_{q^m}$ over ${\mathbb F}_q$, the associated matrix of the quadratic form $f$ is as follows:
$$M_f(\B)=\left(
\begin{array}{ccccc}
		M_{2s-2r+2} & 0 & 0 & 0 & 0 \\
		 0 & 0_{r-2} & K_{r-2} & 0 & 0\\
		 0 & 0 & M_{r-2} & N & 0\\
		 0 & 0 & 0 & M_2 & 0\\
		 0 & 0 & 0 & 0 & 0_{(m-2s)\times (m-2s)}
	\end{array}
\right).$$
It is verified that
$$ M_f^{\prime}(\B)=G^{T}M_f(\B)G=\left(
\begin{array}{ccccc}
		M_{2s-2r+2} & 0 & 0 & 0 & 0 \\
		 0 & 0_{r-2} & I_{r-2} & 0 & 0\\
		 0 & 0 & 0 & (K_{r-2}^{-1})^T N & 0\\
		 0 & 0 & (K_{r-2}^{-1})^T N & M_2 & 0\\
		 0 & 0 & 0 & 0 & 0_{(m-2s)\times (m-2s)}
\end{array}
\right),$$
where $M_2$ is congruent to a matrix of the form
$$ \left(\begin{array}{cc}
		\alpha & 1\\
		0 & \alpha
\end{array}\right), \,\, {\rm and } \,\,
G=\left(
\begin{array}{ccccc}
		I_{2s-2r+2} & 0 & 0 & 0 & 0\\
		0 & I_{r-2} & (K_{r-2}^{-1})^T M_{r-2}^T K_{r-2}^{-1} & (K_{r-2}^{-1})^T N & 0\\
		0 & 0 & K_{r-2}^{-1} & 0 & 0\\
		0 & 0 & 0 & I_2 & 0\\
		0 & 0 & 0 & 0 & 0_{(m-2s)\times (m-2s)}
\end{array}\right).
$$
Moreover, it is easy to show  ${\bf x} M_f^\prime(\B) {\bf x}^T = {\bf x}\bar{M}_f(\B) {\bf x}^T$, where ${\bf x}=(x_1, x_2, \cdots, x_m)$ and
$$ \bar{M}_f(\B)=\left(
\begin{array}{ccccc}
		M_{2s-2r+2} & 0 & 0 & 0 & 0\\
		 0 & 0_{r-2} & I_{r-2} & 0 & 0\\
		 0 & 0 & 0 & 0 & 0\\
		 0 & 0 & 0 & M_2 & 0\\
		 0 & 0 & 0 & 0 & 0_{(m-2s)\times (m-2s)}
\end{array}
\right),
$$
that is to say, the associated matrix of $f$ under the basis $\B$ is $\bar{M}_f(\B)$.

According to the representations of $\bar{H}_r^{\perp}$ and $\bar{H}_{2s-2r+2}$, we know that the restrictions of $\bar{f}$ to $\bar{H}_r^{\perp}$ and $\bar{H}_{2s-2r+2}$ have the same canonical representation. Since $\bar{f}|_{\bar{H}_r^{\perp}}$ is equivalent to a quadratic form on $\bar{H}_r^{\perp}$ of Type II, $\bar{f}|_{\bar{H}_{2s-2r+2}}$ is also equivalent to a quadratic form on $\bar{H}_{2s-2r+2}$ of Type II. So, $M_{2s-2r+2}$ is the associated matrix for a quadratic form of Type II.
By Lemma~11.17 in \cite{Wan2003} and the form of $\bar{M}_f(\B)$, $f$ is equivalent to a quadratic form on ${\mathbb F}_{q^m}$ of Type I.	
In this case, from (\ref{eq:rankofR4}) we know that ${\rm dim}_{\bF_q}\,(\bar{H}_{m-r}\cap \bar{H}_{m-r}^{\bot})=2s-r-\bar{R}=r-2.$
Therefore, the cases $r=1$ and $2\leq r\leq s$ need to be discussed separately.

{\bf Case 1:} $r=1$. From (\ref{neq:R}), the possible values of $\bar{R}$ are $2s-2$ and $2s-1$, i.e., there exists a subspace $H_1\subset \bF_{q^m}$ such that the possible rank of $\bar{f}|_{\bar{H}_1^\perp}$ are $2s-2$ and $2s-1$. By a similar analysis to Theorem~1, we know that $\bar{f}|_{\bar{H}_1^{\perp}}$ is equivalent to the quadratic form of
Type I when $\bar{R}=2s-2$, and Type III when $\bar{R}=2s-1$. Since $f|_{H_1^{\perp}}(x)=\bar{f}|_{\bar{H}_1^{\perp}}(\bar{x})$ for any $x\in {\mathbb F}_{q^m}$ and $R=\bar{R}$, $f|_{H_1^{\perp}}(x)$ is equivalent to Type I when $R=2s-2$ and Type III when $R=2s-1$. So, From (\ref{eq:objvalue1}) we know that $|D_f\cap H_1^{\perp}|$ is maximized when $R=2s-1$. The desirable $H_1$ can be constructed by a similar way in Theorem~\ref{thm:typei}. In this case, $|D_f\cap H_1^{\perp}|=\{x\in {\mathbb F}_{q^m}\,|\,f|_{H_1^{\perp}}(x)=a\}=q^{m-2}$. So,
$$ d_1(\C_{D_f})=n-{\rm max} \left\{ \, |D_f\cap H_1^{\perp}|\, :\, H_1\in \left[ {\mathbb F}_{q^m},1 \right] \right\} = q^{m-1}-q^{m-s-1}-q^{m-2}. $$

{\bf Case 2:} $2\leq r\leq s$. It is shown above that the minimum value of $\bar{R}$ is $2s-2r+2$ when $\bar{f}|_{\bar{H}_r^{\perp}}(\bar{x})$ is equivalent to the quadratic form of Type II.
Since $f|_{H_1^{\perp}}(x)=\bar{f}|_{\bar{H}_1^{\perp}}(\bar{x})$ for any $x\in {\mathbb F}_{q^m}$ and $R=\bar{R}$, we know that the minimum value of $R$ is $2s-2r+2$ when $f|_{H_r^{\perp}}(x)$ is equivalent to
a quadratic form of Type II. So, from (\ref{eq:objvalue1}) we have
$$d_r(\C_{D_f})=n-{\rm max} \left\{ \, |D_f\cap H_r^\perp|\, :\, H_r\in \left[ {\mathbb F}_{q^m},r \right] \right\} =q^{m-1}-q^{m-s-1}-q^{m-r-1}-q^{m-s-2}.$$

{\bf (2)} $s<r<m$. In this case, $0\leq {\rm dim}_{\bF_q}\,(\bar{H}_{r}\cap \bar{H}_{r}^{\bot})\leq {\rm dim}_{\bF_q}\,(\bar{H}_r^{\perp})$. From (\ref{eq:rankofR4}) we have
\begin{equation}\label{eq:barR2}
\bar{R}\geq {\rm type}\, \bar{f}|_{\bar{H}_r^{\perp}}.
\end{equation}
From (\ref{eq:objvalue1}) we know that $|D_f \cap H_r^\perp|$, i.e., the number of solutions for $f|_{H_r^{\perp}}(x)=a$ is maximized when $f|_{H_r^\perp}$ is equivalent to Type II and $R$ is the smallest even number possible.
It is clear that $|D_f \cap H_r^\perp|=0$ when $R=0$, i.e., $f|_{H_r^{\perp}}(x)$ is a zero polynomial. From (\ref{eq:barR2}) we know the second-to-last smallest even number of $\bar{R}$ is $2$. By a similar discussion to
Theorem~\ref{thm:typei}, we can construct a subspace $\bar{H}_r$ such that $\bar{f}|_{\bar{H}_r^{\perp}}$ is equivalent to a quadratic form of Type II and $\bar{R}=$ rank $\bar{f}|_{\bar{H}_r^{\perp}}=2$. In this case,
	$$\bar{R}=R={\rm dim}_{\bF_q}\,(H_r^{\perp})-{\rm dim}_{\bF_q}\,(H_{r}\cap H_{r}^{\bot}),$$
and ${\rm dim}_{\bF_q}\,(H_{r}\cap H_{r}^{\bot})=m-r-2$ for $r\leq m-2$. Therefore, the following cases are discussed separately.
		
{\bf Case 1:} $s<r\leq m-2$. According to the analysis above, $|D_f\cap H_r^\perp|$ is maximized when $f|_{H_r^\perp}$ is equivalent to Type II and $R=2$. From from (\ref{eq:objvalue1}) we have
	 $$d_r(\C_{D_f})=n-{\rm max} \left\{ \, |D_f\cap H_r^\perp|\, :\, H_r\in \left[ {\mathbb F}_{q^m},r\right] \right\} =q^{m-1}-q^{m-s-1}-q^{m-r-1}-q^{m-r-2}.$$
	
{\bf Case 2:} $r=m-1$. For any $(m-1)$-dimensional subspace $H_{m-1}$ of ${\mathbb F}_{q^m}$, ${\rm dim}_{{\mathbb F}_q}\,H_{m-1}^\perp=1$. From (\ref{eq:R}) we know
	$$ R ={\rm rank}\,f|_{H_{m-1}^{\perp}}=1-{\rm dim}_{\bF_q}\,(H_{m-1}^{\perp}\cap H_{m-1})+ {\rm type}\, f|_{H_{m-1}^\perp}. $$
It is known that the number of solutions for the equation $f|_{H_{m-1}^{\perp}}(x)=a$ is 0 when $R=0$. For the case of $R=1$, by a similar discussion to Theorem~\ref{thm:typei}, we can construct a subspace $H_{m-1}^\perp$ of $\bF_{q^m}$ such that $f|_{H_{m-1}^\perp}$ is equivalent to a quadratic form of Type III and $R=$ rank $f|_{H_{m-1}^\perp} =1$. So, from (\ref{eq:objvalue1}) we have
 $$d_{m-1}(\C_{D_f})=n-{\rm max} \left\{ \, |D_f\cap H_{m-1}^\perp\, :\, H_{m-1}\in \left[ {\mathbb F}_{q^m},m-1\right] \right\} =q^{m-1}-q^{m-s-1}-1.$$
	
{\bf Case 3:} $r=m$. From (\ref{eq:rankofR4}) we know $\bar{R}=R={\rm rank}\,f|_{H_m^{\perp}}=0$. So, $|D_f\cap H_m^\perp|=0$ and
	 $$d_m(\C_{D_f})=q^{m-1}-q^{m-s-1}.$$
\end{proof}

\begin{example}
Let $w$ be a primitive element of ${\mathbb F}_{2^6}$ and $f(x)={\rm Tr}_1^6( w^3x^3)$ be a quadratic form on ${\mathbb F}_{2^6}$, where ${\rm Tr}_1^6(\cdot)$ is a trace function from ${\mathbb F}_{2^6}$ to $\bF_2$. Let $\C_{D_f}$ be a linear code as in (\ref{eq:defcode}), where $D_f=\{ x\in \bF_{2^6} \, |\, f(x)=1 \}$.
By the help of Magma, we obtain the weight hierarchy of $\C_{D_f}$ as follows: $d_1=8, d_2=12, d_3=18, d_4=21, d_5=23,d_6=24$. This result is consistent with Theorem~\ref{thm:iv}.	
\end{example}

By a proof similar to Theorem \ref{thm:iv}, we get the following theorem.
\begin{theorem}\label{thm:v}
Let $m$ be a positive integer and $f$ be a degenerate quadratic form over ${\mathbb F}_{q^m}$ with $rank\ f=2s\,(2s<m)$, which is equivalent to Type \uppercase\expandafter{\romannumeral2}.  Then the linear code $\C_{D_f}$ defined in (\ref{eq:defcode}) has the following weight hierarchy:
	\begin{equation*}
	 	d_r(\C_{D_f})=\begin{cases}
			q^{m-1}+q^{s-1}-q^{m-r-1}-q^{m-s-1}, &{\rm if}\ 1\leq r\leq s-1,\\
			q^{m-1}+q^{s-1}-q^{m-r-1}-q^{m-r-2}, &{\rm if}\ s\leq r<m-1,\\
			q^{m-1}+q^{s-1}-1, &{\rm if}\ r=m-1,\\
			q^{m-1}+q^{s-1}, &{\rm if}\ r=m.
		\end{cases}
	 \end{equation*}
\end{theorem}

\begin{example}
Let $w$ be a primitive element of ${\mathbb F}_{2^4}$ and $f(x)={\rm Tr}_1^4( w^3x^3)$ be a quadratic form on ${\mathbb F}_{2^4}$, where ${\rm Tr}_1^4(\cdot)$ is a trace function from ${\mathbb F}_{2^4}$ to $\bF_2$. Let $\C_{D_f}$ be a linear code as in (\ref{eq:defcode}), where $D_f=\{ x\in \bF_{2^4} \, |\, f(x)=1 \}$.
By the help of Magma, we obtain the weight hierarchy of $\C_{D_f}$ as follows: $d_1=6, d_2=9, d_3=11, d_4=12$. This result is consistent with Theorem~\ref{thm:v}.	
\end{example}

By a discussion similar to Theorem \ref{thm:iv}, we obtain the following theorem.
\begin{theorem}\label{thm:vi}
Let $m$ be a positive integer and and $f$ be a degenerate quadratic form over ${\mathbb F}_{q^m}$ with $rank\ f=2s+1\,(2s+1<m)$, which is equivalent to Type \uppercase\expandafter{\romannumeral3}.  Then the linear code $\C_{D_f}$ defined in~(\ref{eq:defcode}) has the following weight hierarchy:
	\begin{equation*}
	 	d_r(\C_{D_f})=\begin{cases}
			q^{m-1}-q^{m-r-1}-q^{m-\frac{2s+3}{2}}, &{\rm if}\ 1\leq r\leq s,\\
			q^{m-1}-q^{m-r-1}-q^{m-r-2}, &{\rm if}\ s+1\leq r<m-1,\\
			q^{m-1}-1, &{\rm if}\ r=m-1,\\
			q^{m-1}, &{\rm if}\ r=m.
		\end{cases}
	 \end{equation*}
\end{theorem}

\begin{example}
Let $w$ be a primitive element of ${\mathbb F}_{2^6}$ and $f(x)={\rm Tr}_1^6( wx^5)$ be a quadratic form on ${\mathbb F}_{2^6}$, where ${\rm Tr}_1^6(\cdot)$ is a trace function from ${\mathbb F}_{2^6}$ to $\bF_2$. Let $\C_{D_f}$ be a linear code as in (\ref{eq:defcode}), where $D_f=\{ x\in \bF_{2^6} \, |\, f(x)=1 \}$.
By the help of Magma, we obtain the weight hierarchy of $\C_{D_f}$ as follows: $d_1=12, d_2=20, d_3=26, d_4=29,d_5=31,d_6=32$. This result is consistent with Theorem~\ref{thm:vi}.	
\end{example}

\section{Concluding remarks}\label{sec5}

Quadratic forms on $\bF_{q^m}$ behave quite differently depending on whether the characteristic of $\bF_q$ is $2$. Some additional information is needed to classify
quadratic forms over finite fields of even characteristic. For a quadratic form $f$ over $\bF_{q^m}$, where $q$ is a power of $2$, by carefully studying the behavior of the quadratic form $f$ restricted to subspaces of $\bF_{q^m}$, we obtained the number of solutions for the restricted quadratic equation $f|_{H}(x)=a$, where $a\in \bF_{q^m}^*$ and $f|_{H}$ is the restriction of $f$ to a subspace $H\subset \bF_{q^m}$. Based on this result, we determined completely the weight hierarchies of linear codes from quadratic forms over finite fields of even characteristic. Our results complement the results in~\cite{Li2021,Li2022}, and the weight hierarchies of linear codes from quadratic forms were completely determined. By the help
of Magma, we gave some numerical examples, which verified the correctness of our theorems.

\end{document}